%% file: book-skeleton.tex
  \documentclass{EngC}
  \usepackage{framed}         
  \usepackage{soul}           
\usepackage[redeflists]{IEEEtrantools}

  \usepackage{rotating}
  \usepackage{floatpag}
  \rotfloatpagestyle{empty}

  \usepackage{amsmath}
  \usepackage{amsfonts}
  \usepackage{amsthm}
  \usepackage{mathrsfs}
  \usepackage{graphicx}
  \usepackage{subcaption}

  \usepackage{multind}
  \makeindex{authors}
  \makeindex{subject}

  \theoremstyle{plain}
  \newtheorem{theorem}{Theorem}[chapter]

  \theoremstyle{definition}
  
   \newtheorem{example}[theorem]{Example}

  \theoremstyle{remark}

\graphicspath{ {./coded/figures/}}

\begin{document}


  \include{notation}

  \include{coded/coded-chapter}



  \bibliographystyle{IEEEtran} 

  \bibliography{coded/coded}\label{refs}






\end{document}

%% file: coded/coded-chapter.tex
\newcommand{\codedsetupfigsscale}{.6}

\chapter{Coded Caching for Heterogeneous Wireless Networks
}
\vspace{-4em}
 \textit{\Large Jad Hachem, Nikhil Karamchandani, Suhas Diggavi, and Sharayu Moharir}
\label{chap:coded-caching}

\section{Introduction}
\label{coded:sec:intro}
\input{coded/sections/intro.tex}

\section{Overview of Coded Caching}
\label{coded:sec:overview}
\input{coded/sections/overview.tex}

\section{Non-Uniform Content Popularity}
\label{coded:sec:popularity}
\input{coded/sections/popularity.tex}

\section{Multiple Cache Access}
\label{coded:sec:multi-access}
\input{coded/sections/multi-access.tex}

\section{Network Structure}
\label{coded:sec:gen}
\input{coded/sections/Generalizations.tex}

%% file: coded/sections/intro.tex
%

\begin{figure}
\centering
\includegraphics[width=.7\textwidth]{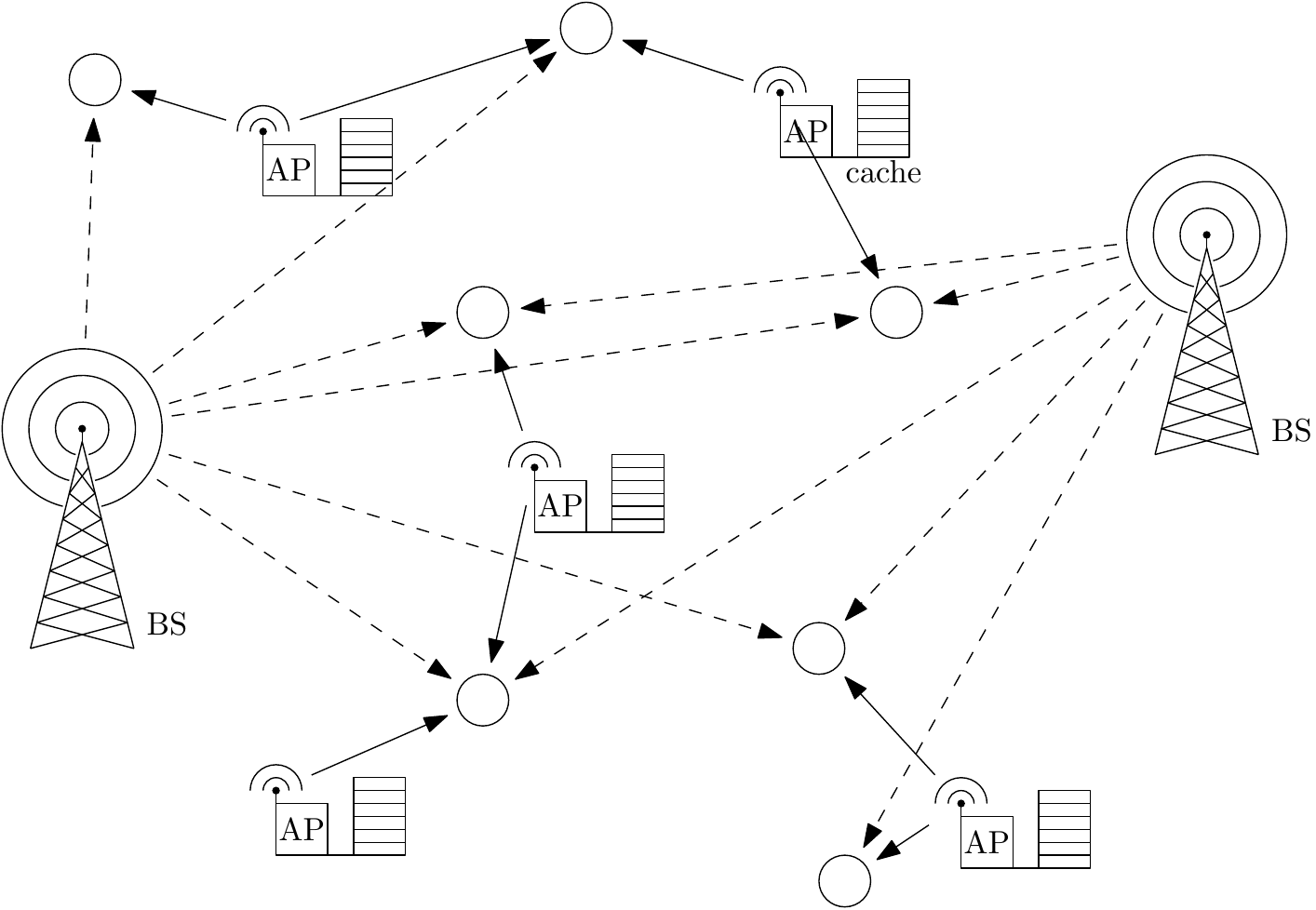}
\caption{Caching in a wireless heterogeneous network (HetNet).}
\label{coded:fig:hetnet}
\end{figure}

Broadband data consumption has witnessed a tremendous growth over the
past few years, due in large part to multi-media applications such as
Video-on-Demand. Increasing data demand has been managed in wired internet via Content Distribution Networks (CDNs) that mirror data in various locations and effectively push content
closer to end users. CDNs help reduce host server load by serving user
requests locally via content cached locally. This solution works best
when neither local storage nor data rates are
bottlenecks \cite{korupolu1999}. Neither of these is
true in cellular networks; the last-hop wireless link has low throughput
(improvements in cellular data
rates do not sufficiently
compensate for exploding in demand) and there is virtually no storage
at base-stations. To address the throughput issue, a heterogeneous wireless network (HetNet) architecture
has been proposed for 5G systems \cite{QualcommSmallCells, IntelHetNet}. HetNets consist of a dense deployment of very small
cells (pico/femto) with high data rates, combined with a sparse
deployment of larger macro-cellular base stations (BSs) of
comparatively lower data rates; WiFi access points (APs) can be a
typical small cell. However, this architecture is `incomplete' because
the APs are connected to the backbone via best-effort backhaul, which
is a bottleneck \cite{IntelSmallCells}. And even joint management of
APs and BSs cannot provide enough improvement to deal with projected demand
growth \cite{QualcommSmallCells,CiscoReport}. This leads us to argue
that the traditional CDN approach in an enhanced wireless system
design is an incomplete solution---the CDNs optimize content placement
without accounting for characteristics of wireless communications and
wireless system design only focuses on increasing delivery rates,
agnostic to content. To fully enable content-centric wireless
networks, we need a  joint design of content placement, access, and
delivery. Broadly, the proposal is to provide
storage capabilities at the network nodes (base-stations and WiFi
access-points) and create a large-scale distributed cache. Users will
be served by connecting them to one or more nodes hosting their
requested content. Delivery protocols will use algorithms that are
aware of attributes of wireless networks like the broadcast medium and 
interference. Figure~\ref{coded:fig:hetnet} illustrates this.

In this chapter we describe a problem based on an architecture where
content is stored at multiple APs without \emph{a priori} knowing the user
requests, and the base-station broadcast is used judiciously to
complement the local caching, \emph{after} the user requests are
known. This is motivated by the new approach initiated
in the seminal works \cite{maddah-ali2012, maddah2015decentralized}, where it has been shown that joint
design of storage and delivery (a.k.a. ``coded caching'') can
significantly improve content delivery rate requirements. This was
enabled by content placement that creates (network-coded) multicast
opportunities among users with access to different storage units, even
when they have different (and \emph{a priori} unknown) requests. This enables
an examination of the optimal trade-off between the cache memory size and the 
broadcast delivery rate.

We will begin by discussing the setup studied in \cite{maddah-ali2012,maddah2015decentralized} which introduced the idea of coded caching and describing their main results. These works considered the case where all files in the catalogue have the same popularity.  However, it is well understood that content
demand is non-uniform in practice, with some files being more popular
than others.  Motivated by
this, we describe models that take this non-uniform popularity into account and discuss how it impacts the results and proposed caching and delivery schemes.
While all the above focus on a setup where each user has (fixed) access to a single unique cache and hears the common message broadcast by the base station, several generalizations to the network structure have been studied recently.
In particular, we discuss in detail the case where, based on the user requests, each user can be adaptively matched to one cache (possibly amongst a subset of caches). Finally, we end the chapter with some further generalizations of the problem that have been studied in the literature.

%% file: coded/sections/overview.tex
Coded caching was first introduced in 2012 by Maddah-Ali and Niesen \cite{maddah-ali2012} as a solution to the content distribution problem in a wireless setting.
In order to focus on this new technique, the setup ignored variations in content popularity and limited user-to-cache access to exactly one user connecting to one cache.
The authors showed that conventional caching techniques are inefficient in such a setup.
Instead, one can leverage the broadcast capabilities inherent to wireless communications in order to send a small network-coded message that can serve a large number of users at once.
The setup and ideas became fundamental to much of the following literature on the subject, and so in this chapter we give an overview of the results and insights from \cite{maddah-ali2012}.

\subsection{Setup and Notation}
\label{sec:setupbasic}
We begin by describing the setup studied in the seminal work of Maddah-Ali and Niesen \cite{maddah-ali2012}; see Figure~\ref{coded:fig:setup-basic} for an illustration of the system.
Consider a server hosting a content library with $N$ files, labeled $W_1,\ldots,W_N$, of size $F$ bits each.
There are $K$ users in the network, each of which is equipped with a local cache of size $MF$ bits.
The server is connected to the users via an error-free broadcast link.

\begin{figure}
\centering
\includegraphics[scale=\codedsetupfigsscale]{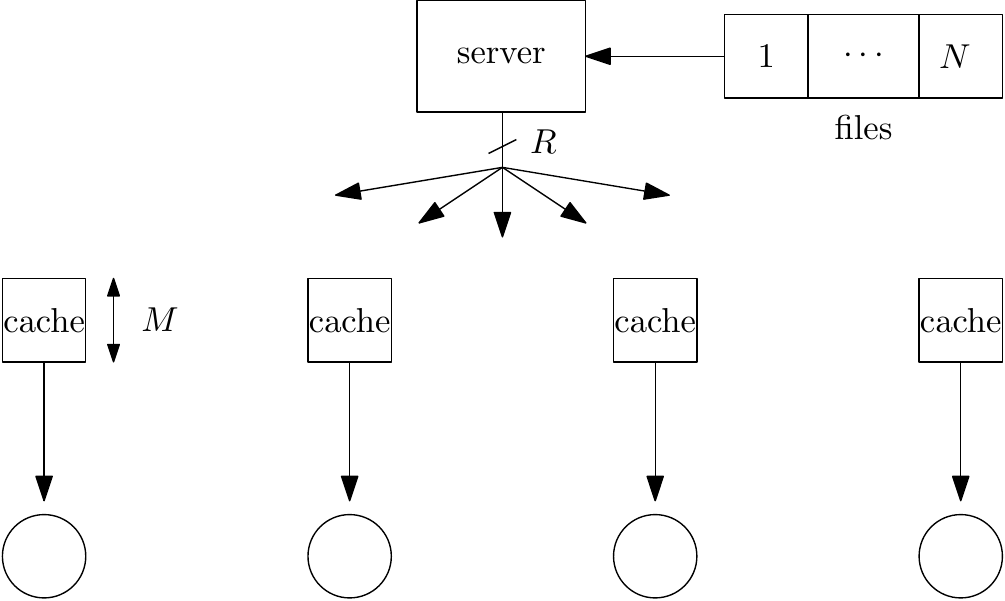}
\caption{The basic coded caching problem, with $N$ files, $K$ caches, and one user at every cache.
There is an error-free broadcast link from the server to the users.}
\label{coded:fig:setup-basic}
\end{figure}

The system operates in two phases.
We start with a \emph{placement phase} in which all the user caches are populated with content related to the $N$ files.
No restrictions are posed on this placement phase aside from the cache memory constraint.
In particular, the caches are not restricted to holding just files or parts of files; any function (deterministic or not) of the files can be used, and the caches do not necessarily have to hold the same information.
Crucially, this is done before the user requests are revealed to the system.
Next, we move to the \emph{delivery phase}, which starts with each user making a request for one file from the content library.
Based on the user requests and the content stored in the caches, the server transmits a  message, across the error-free broadcast link, of size $RF$ bits intended to serve the requests of all the users.
Each user then combines this message with the contents of its own cache in order to recover the file that it requested.

Our resources here are the \emph{cache memory} $M$ and the \emph{broadcast rate} $R$.
Clearly there is a trade-off between them: the larger the cache, the smaller the size of the broadcast message needed to serve the users.
The goal is to characterize the optimal trade-off.
Formally, we say that a pair $(R,M)$ is achievable if there exists a caching scheme $\mathcal{S}_F$ for every file size $F$ such that:
\begin{itemize}
\item $\mathcal{S}_F$ uses a cache memory at each user of capacity at most $MF$ bits and a broadcast rate from the server of size at most $RF$ bits; and
\item For any collection of user requests, the probability that each user recovers its requested file without error goes to one as $F\to\infty$.
\end{itemize}
Then, our goal is to find, for every $M\ge0$, the information-theoretically optimal rate defined as
\begin{equation}
\label{coded:eq:optimal-rate}
R^\star(M) = \inf \left\{ R : (R,M) \text{ is achievable} \right\}.
\end{equation}
Note that the infimum in \eqref{coded:eq:optimal-rate} is over all possible caching and delivery schemes, without any restrictions.

We discuss a small example below to illustrate the setup as well as some representative caching and delivery schemes.

\subsection{A Small Illustrative Example}
\label{coded:sec:overview:example}
Consider a special case of the system described above with $N = 2$ files, $K = 2$ users, and $M = 1$ memory at each user.
Denote the two files by $A$ and $B$.
Any caching and delivery scheme has to specify what to store in the caches during the placement phase and what the server should transmit during the delivery phase so that the user requests can be served.

For instance, a natural cache placement strategy is to split file $A$ into two equal parts $A_1,A_2$ and similarly file $B$ into $B_1,B_2$.
Each cache stores one half of each file.
In a conventional caching and delivery system, each cache would store $(A_1,B_1)$ and the server would handle each user request via a separate transmission.
For example, if one user wants file $A$ and the other wants file $B$, the server would transmit $A_2$ to the first user and $B_2$ to the second user, and thus the broadcast message size is equivalent to the size of one file.
This scheme leverages the local presence of a cache at every user: each user has access to the half-file present in its cache, which reduces the message size by that amount per user.
However, using network coding techniques, we can design the cache contents in such a way that each user benefits from the contents of both its cache \emph{and the other user's cache}.

To do so, we consider an alternate placement and delivery strategy whose main idea is to store different file parts in each user's cache in a way that enables sending linear combinations of file parts that are simultaneously useful to both users.
Firstly, in the placement phase the first user's cache stores $(A_1,B_1)$ while the second user's cache stores $(A_2,B_2)$.
Secondly, instead of treating the two user requests separately during the delivery phase, we consider them jointly.
For example, if the first user requests file $A$ and the second user requests file $B$, then the server sends a linear combination $A_2\oplus B_1$ on the shared broadcast link, where $\oplus$ denotes the bitwise-XOR operation.
The first user has $A_1$ available in its local cache and can combine $B_1$ with $A_2\oplus B_1$ in order to recover $A_2$; the second user can similarly obtain both $B_1$ and $B_2$.
Thus the users' requests were both served with a broadcast message of size only half a file.

Note that the above scheme reduces the server transmission rate by a factor of two over a conventional scheme.
The main idea is to carefully design the cache contents so as to maximize the number of coded multicasting opportunities during server transmission, enabling the server to send a single message satisfying multiple users, possibly requesting different files, simultaneously.

\subsection{Achievable rate}
\label{sec:centscheme}
The above ideas were generalized in \cite{maddah-ali2012}, which proposed a new caching and delivery scheme for the general setup described in Section~\ref{sec:setupbasic} and also characterized its achievable rate, as shown in the following result.
%
\begin{theorem}
\label{coded:thm:singlelevel-achieve}
For the system described in Section~\ref{sec:setupbasic} with $M\in\frac{N}{K}\cdot\{0,1,\ldots,K\}$, there exists a placement and delivery scheme which achieves the following server transmission rate:
\[
R(M) = K \cdot \left( 1 - \frac{M}{N} \right) \cdot \frac{1}{1+KM/N}.
\]
For $M\in[0,N]$, the lower convex envelope of these points can be achieved.
\end{theorem}

A rate of $N-M$ can also be achieved (without coded caching) and is useful when the number of files is small, but in the more relevant case where $N\ge K$, the rate in Theorem~\ref{coded:thm:singlelevel-achieve} is smaller and we will henceforth focus on it.

The scheme achieving this rate is a generalization of the ideas described in Section~\ref{coded:sec:overview:example}, and we describe it below.
Before we do that, we can gain some insights about the achievable rate in Theorem~\ref{coded:thm:singlelevel-achieve} by factoring it into three terms.
The first term, $K$, is the total number of users and represents the rate needed without caching, since in the worst case the server might be required to  transmit $K$ distinct files.
The second term, $1-M/N$, is referred to as the \emph{local caching gain}.
It is the gain obtained by the fact that each user already has a fraction $M/N$ of its requested file stored locally.
The third term, $1/(1+KM/N)$, is referred to as the \emph{global caching gain}.
It is specifically achieved by the fact that the server sends coded multicast messages that are useful to many users at once (more precisely, each bit is used by exactly $1+KM/N$ users).
In effect, the coded multicast allows each user to benefit from the caches of all the other users as well, hence the appearance of the total system memory $KM$ in the expression.
This is the gain that the proposed scheme derives over a conventional caching and delivery scheme which serves the user requests through separate unicast server messages.

The significance of the global caching gain can be captured by noticing that the achievable rate in Theorem~\ref{coded:thm:singlelevel-achieve} can be upper-bounded by
\begin{equation}
\label{coded:eq:singlelevel-ub}
R(M) \le \min\left\{K, \frac{N}{M}\right\}\left( 1 - \frac{M}{N} \right).
\end{equation}
Consequently, as long as the total memory in the network is large enough to hold the entire library (i.e., $KM\ge N$),  the achievable rate is at most $N/M-1$, and is thus \emph{independent of the number of users}!

\begin{proof}[Proof of Theorem~\ref{coded:thm:singlelevel-achieve}] We now describe a placement and delivery scheme for the general system described in Section~\ref{sec:setupbasic}.

\emph{Placement phase}: Denote the files by $W_1, W_2, \ldots, W_N$.
Let $t \overset{\Delta}{=} MK/N$.
From the statement of the theorem, note that $t$ is an integer between $0$ and $K$.
Divide each file $W_i$ into $K \choose t$ equal parts and index them as follows:
$$
W_i = \left(W_i^{S} : S \subseteq \{1,2,\ldots,K\}, |S| = t\right).
$$
Note that each subfile is of size $F / {K\choose t}$.
For any user $i$, its cache stores ${K-1\choose t - 1}$ pieces of each file $W_j$ given by
$$
\left(W_j^{S}: i \in S, S \subseteq \{1,2,\ldots,K\}, |S| = t\right).
$$
The total amount of storage each cache needs to store these pieces is given by
$$
N \cdot {K- 1 \choose t - 1} \cdot \frac{NF}{{K\choose t}} = \frac{Ft}{K} = MF
$$
where the last equality follows from the definition of $t$.
Thus, the storage constraint is satisfied at each cache.

\emph{Delivery phase}: For $t = K$, the memory at each cache is $M = N$ and is sufficient to store the entire file catalogue.
Thus, the required server transmission rate is zero.
Below we consider $t \in \{0,1,2,\ldots,K-1\}$.
Denote by $d_i$ the index of the file requested by user $i$, i.e. user $i$ requests file $W_{d_i}$.
Consider a subset $S \subseteq \{1,2,\ldots,K\}$ of size $t + 1$.
Note that for each $j\in S$, there is a subfile of its requested file $W_{d_j}$, given by $W_{d_j}^{S\backslash \{j\}}$, which is stored in the caches at all the other users in $S$.
Corresponding to this subset $S$, the server transmits the message
\begin{equation}
\label{Eqn:transmittedmsg}
\oplus_{j \in S}  W_{d_j}^{S\backslash \{j\}}.
\end{equation}
We repeat the above procedure for each of the ${K\choose t+1}$ subsets of $\{1,2,\ldots,K\}$ of size $t + 1$.

We now show that the above placement and delivery scheme allows each user $i$ to recover its requested file $W_{d_i}$.
Consider a  subset $S \subseteq \{1,2,\ldots,K\}$ of size $t + 1$ such that $i \in S$.
Note from \eqref{Eqn:transmittedmsg} that amongst the $t+1$ subfiles involved in the transmitted message corresponding to $S$, all except its desired subfile $W_{d_i}^{S\backslash \{i\}}$ is already available in the cache of user $i$.
Thus, the user is able to recover its desired subfile.

Repeating the above argument, user $i$ is able to recover all the subfiles of the form
$$
\left(W_{d_i}^{T} : T \subseteq \{1,2,\ldots,K\} \backslash \{i\}, |T| = t\right).
$$
Furthermore, all the other subfiles of the requested file $W_{d_i}$ are already available in the cache of user $i$.
Thus, the above described placement and delivery strategy represents a feasible scheme for the general system.

To complete the proof of the theorem, we have to evaluate the server transmission rate for the proposed scheme.
From \eqref{Eqn:transmittedmsg}, the size of the transmission corresponding to any subset $S$ of size $t+1$ is $F / {K\choose t}$.
Since there is one such transmission corresponding to each such subset, the total transmission size is given by
$$
{K\choose t + 1} \cdot \frac{F}{{K\choose t}} = \frac{K - t}{t+1} = K \cdot \left( 1 - \frac{M}{N} \right) \cdot \frac{1}{1+KM/N}.
$$
\end{proof}

\paragraph{Decentralized scheme}
Note that the above described scheme carefully orchestrates the placement phase to create simultaneous coded multicasting opportunities during the delivery phase.
In particular, the number of users and their identities are required to ensure that each cache is populated with the right file pieces.
Since this kind of information might not always be available in practice, the authors develop a decentralized placement (and corresponding delivery scheme) scheme in \cite{maddah2015decentralized} where each user randomly samples $MF$ bits from the $NF$ bits in the content library.
The achievable rate of this scheme is characterized in \cite{maddah2015decentralized} and is presented below:
\begin{theorem}
\label{deccoded:thm:singlelevel-achieve}
Consider the system described in Section~\ref{sec:setupbasic}.
For $M\in[0,N]$, there exists a decentralized placement scheme and a corresponding delivery scheme which achieves a server transmission rate arbitrarily close to
\[
R_D(M) = K \cdot \left( 1 - \frac{M}{N} \right) \cdot \min\left\{\frac{N}{KM} \left(1 - (1 - M/N)^K\right), \frac{N}{K}\right\}
\]
for a large enough file size $F$.
\end{theorem}

Although the decentralized scheme cannot control the placement as precisely as the centralized scheme, the authors nevertheless show that the decentralized placement creates almost as many coding opportunities as the centralized placement with very high probability.
In fact, the resulting achievable rates are within a constant multiplicative factor of each other.

\subsection{Approximate Optimality}
\label{sec:approxopt}
Next, we examine how the performance of the centralized placement and delivery scheme proposed in Section~\ref{sec:centscheme}  compares to the optimal scheme for this setup with no restrictions on the placement and delivery phases.
The next theorem \cite{maddah-ali2012} states the approximate optimality of the achievable rate $R(M)$ in Theorem~\ref{coded:thm:singlelevel-achieve} with respect to the optimal rate $R^\star(M)$ as defined in \eqref{coded:eq:optimal-rate}.
\begin{theorem}
\label{coded:thm:singlelevel-orderoptimal}
The rate achieved in Theorem~\ref{coded:thm:singlelevel-achieve} is within a constant multiplicative factor of the optimal rate.
Specifically, for all values of $N$, $K$, and $M$,
\[
1 \le \frac{R(M)}{R^\star(M)} \le 12.
\]
\end{theorem}
Note that the bound is independent of the problem parameters: the achievable rate is within a factor of 12 of the optimum even if $N$ and $K$ are arbitrarily large.
Proving Theorem~\ref{coded:thm:singlelevel-orderoptimal} requires deriving information-theoretic lower bounds on the optimal rate using cut-set based arguments \cite{BookThomasCover}.
We will not cover this here and instead point the interested reader to \cite{maddah-ali2012} for details.
Finally, while the constant gap factor of $12$ is indeed quite large, there have been significant improvements in terms of both the achievable rates \cite{amiri2017fundamental, zhang2018fundamental,wei2017novel} as well as the lower bound arguments \cite{wang2016new, wang2017improved, ghasemi2017improved} which can be used to tighten the gap significantly.
In fact, under the restriction of uncoded placement, the rate proposed in Theorem~\ref{coded:thm:singlelevel-achieve} is shown to be exactly optimal in \cite{yu2018exact, wan2016optimality}.


%% file: coded/sections/popularity.tex
By studying the worst-case rate over all possible user demands, the problem in \cite{maddah-ali2012} effectively ignores any content popularity, since in practice some files can be requested more frequently than others.
One way to incorporate content popularity into the problem is by setting the user requests to be stochastic, following some probability distribution, and then analyzing the expected broadcast rate.
A common distribution to model content popularity is the Zipf distribution, which is widely observed for many content libraries such as the YouTube video catalogue \cite{YoutubeRepository}.
Coded caching is studied under such a distribution in \cite{JiZipf}, and the approximately optimal expected server transmission rate is characterized.
In \cite{niesen2017coded, ZhangArbitrary}, the case of arbitrary popularity distributions is studied and the approximately optimal expected rate is derived.

In general, the popularity of a file can be thought of as the likelihood that a given user will request this file.
Under a stochastic popularity model, this translates to a probability distribution over the files such that each user requests one file based on this probability.
Note that, since the number of files is typically large compared to the number of users, we cannot reliably predict the number of users requesting each file from prior requests, especially for the less popular files.
However, if the files are partitioned into a small number of levels by grouping together contents of similar popularity, we can more reliably estimate the cumulative popularity across these levels.
If the number of users is large compared to the number of levels, the number of users \emph{per level} under the stochastic popularity model will concentrate around the average value.
The multi-level popularity model, introduced in \cite{HKDmultilevel}, captures this aspect by making the number of users requesting files from each level fixed, deterministic, and known \emph{a priori}; the worst-case rate (under this restriction on the demand) is then analyzed in a similar vein as in Section~\ref{sec:setupbasic}.
We discuss this multi-level popularity model in this section.

More formally, in the multi-level popularity model, the files in the content library are partitioned into a certain number of groups called popularity levels.
Each level $i\in\{1,\ldots,L\}$ consists of $N_i$ files, and there are a total of $K_i$ users requesting files from this level.
Each of these $K_i$ users can request any file belonging to level $i$.
It is useful to think of the popularity of each file in level $i$ as being proportional to the number of users per file of the level, $K_i/N_i$.
For simplicity, we restrict the discussion here to the case where there are more files than users for every level, i.e., $N_i\ge K_i$ for all $i\in\{1,\ldots,L\}$.
Note that the setup studied in Section~\ref{coded:sec:overview} is a special case with $L =1$ level, $N_1 = N$, and $K_1 = K$.

The multi-level popularity model turns out to be useful in studying how the total number of users in the network, as compared to the number of caches, affects the system under non-uniform popularity.
We will look at two extremes: one in which each cache has exactly one associated user (the single-user setup), and one in which each cache has a large number of associated users (the multi-user setup).
In the single-user setup, only one level is represented at each cache since each user requests one file from one popularity level, as shown in \figurename~\ref{coded:fig:setup-ml-su}.
In the multi-user setup, the number of users is large enough for every level to be represented at every cache by at least one user, as shown in \figurename~\ref{coded:fig:setup-ml-mu}.
Interestingly, it turns out that the strategies required for these two setups are quite different: a level-merging approach works for the single-user setup, while a level-separation approach works for the multi-level setup.

\begin{figure}
\begin{subfigure}{\textwidth}
\centering
\includegraphics[scale=\codedsetupfigsscale]{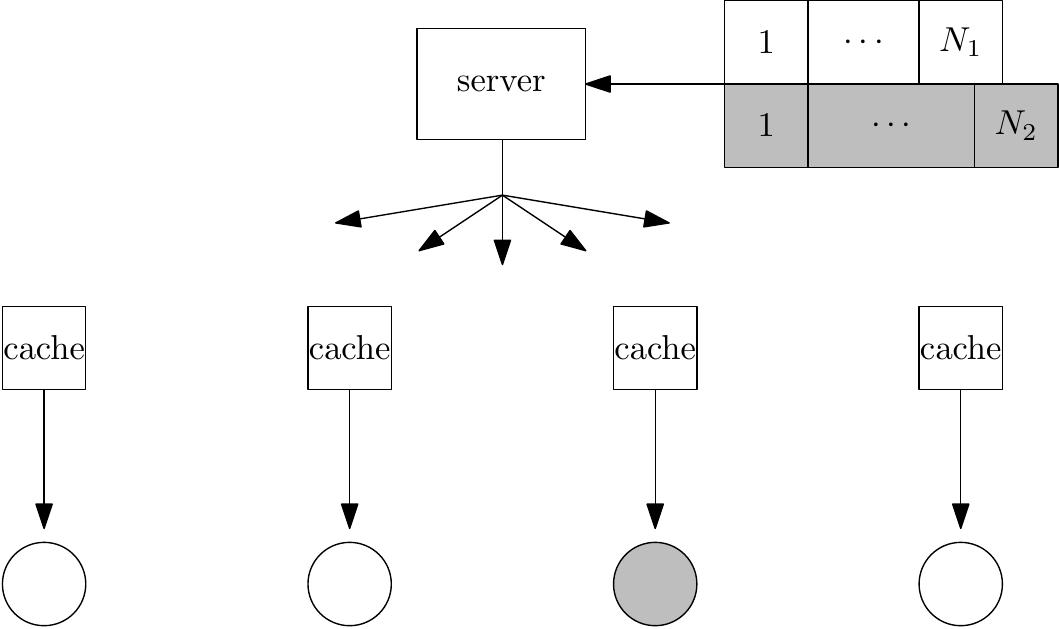}
\caption{The single-user setup.}
\label{coded:fig:setup-ml-su}
\end{subfigure}

\begin{subfigure}{\textwidth}
\centering
\includegraphics[scale=\codedsetupfigsscale]{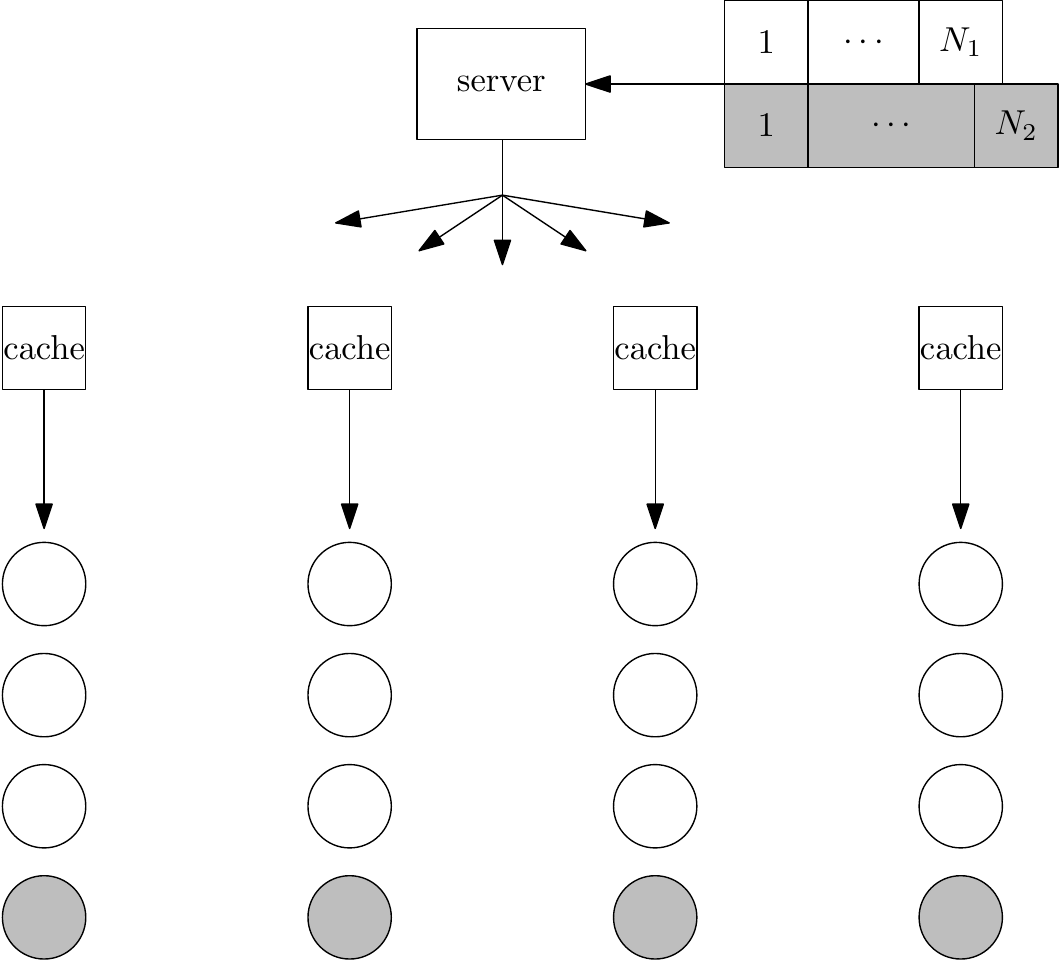}
\caption{The multi-user setup.}
\label{coded:fig:setup-ml-mu}
\end{subfigure}
\caption{Multi-level popularity, with $L=2$ levels.}
\end{figure}

In order to understand the major difference between these two setups, it is useful to first reflect on what enables the coding gains in the original setup in Section~\ref{coded:sec:overview}.
Recall that a coded message from the server consists of a linear combination of parts of files requested by a subset of users, for example see \eqref{Eqn:transmittedmsg}.
Each such user has access to \emph{different} side-information through their distinct caches.
It is precisely this difference in side-information that allows the same linear combination to be beneficial to multiple users possibly requesting distinct files.
If two users' cache contents were identical, then their side-information is identical and no coding gains can be achieved among them.

Moreover, because of the symmetry of the delivery message in \eqref{Eqn:transmittedmsg}, the procedure is most efficient when the involved subfiles are of the same size; or equivalently when the files involved are stored equally in each cache.
When adapting this to the non-uniform popularities setup, this creates a conflict: since the more popular files are more likely to be requested, we would want to give them a larger portion of the cache memory than the less popular files.
However, as mentioned before, this would negatively impact the efficiency of coded messages involving requests for files with significantly different popularities.
Hence there is potentially a dilemma between giving the more popular files a larger memory share on the one hand, and obtaining more efficient coding opportunities on the other hand.

The dilemma is naturally resolved for the multi-user setup.
Notice from  \figurename~\ref{coded:fig:setup-ml-mu} that the users can be partitioned into ``rows'' of $K$ users each, such that each user in the row connects to one cache and requests a file from the same popularity  level.
A natural strategy here is to ensure that each linear combination the server sends is intended only for a subset of users belonging to the same row.
Since all involved requests in a row will be for files belonging to the same level, they will have the same popularity and hence the same allocated cache memory.
Furthermore, grouping an additional  user (requesting a distinct file) with a row of $K$ users in a single coded-multicast transmission cannot be beneficial since this user will necessarily share a cache with another user in the considered row.
These two users will thus have access to the same side-information, and hence as discussed before no coding gains can be obtained between them.
As we will see later, the approximately optimal strategy here is to partition each cache amongst the various levels during the placement phase and then address the demands in each row of users separately during the delivery phase using coded-multicast transmissions as discussed in Section~\ref{coded:sec:overview}.

On the other hand, the dilemma is not so easily resolved in the single-user setup.
Notice from \figurename~\ref{coded:fig:setup-ml-su} that in this case there is only one ``row'' of users in which all the file popularity levels are represented.
This is unlike the multi-user setup where all users in a row requested files from the same popularity level, and hence if we allow all linear combinations in the server transmission, we might have to combine requests for files with very different popularities.
However, if we restrict server transmissions to combine only requests belonging to the same popularity level, that will limit the coded-multicasting opportunities severely and increase the required server transmission rate.
As we will see later, it turns out that the approximately optimal strategy here is to ``merge'' a subset of the higher popularity levels so that all the files belonging to them are given the same amount of memory, and so that the requests belonging to these levels can be efficiently combined in the coded-multicast messages.

Thus, the schemes corresponding to the multi-user and single-user setups have different philosophies, and this difference marks the dichotomy between the two setups.
Next, we study each of these setups in more detail.

\subsection{The Single-User Setup}

In the single-user setup, there is exactly one user connected to each of the $K$ caches.
As mentioned before, $K_i$ users  in the system request  a file from level $i$ and $K_1+\cdots+K_L=K$.
Importantly, while placing content in the caches, we know exactly how many users will request a file from each level, but we do not know \emph{which} users will request from which level.

As discussed before, the idea in this setup is to strike a balance between two opposing principles: creating coding opportunities across popularity levels on the one hand, and allocating more of the cache memory to the more popular files on the other hand.
The balance that turns out to be approximately optimal is to partition the levels into two groups, which we will call $H$ and $I$.
The files belonging to levels in the set $I$ will all be treated as if they are of the same popularity and are all allocated the same amount of memory; effectively the levels in $I$ are merged into one super-level with $\sum_{i\in I}N_i$ files and $\sum_{i\in I}K_i$ users.
All the cache memory will be given to the set $I$, while the files in set $H$ will not be stored at all.
A coded caching scheme is then used on the set $I$ as described in Section~\ref{sec:centscheme}, and all requests for files from set $H$ are handled by direct unicast transmissions from the server.

We can therefore apply Theorem~\ref{coded:thm:singlelevel-achieve} on each of $H$ and $I$ separately, which using \eqref{coded:eq:singlelevel-ub} yields an achievable rate upper-bounded by
\begin{equation}
\label{Eqn:AchievableRateSingleUserPerCache}
R_\mathrm{SU}(M) \le \max\left\{\frac{\sum_{i\in I}N_i}{M} - 1, 0\right\} + \sum_{h\in H}K_h.
\end{equation}
The maximization with zero is necessary since, depending on the choice of $I$, the memory $M$ could be larger than $\sum_{i\in I}N_i$.

\begin{example}
Consider an example multi-level single-user setup with $L=3$ file popularity levels, $N_1 = 100, N_2 = 500, N_3 = 1000$ files and $K_1 = 100, K_2 = 50, K_3 = 5$ users.
Consider the memory per cache to be $M = N_1 = 100$.
We evaluate the rate of the above proposed strategy for different choices of $H$, $I$:
\begin{enumerate}
\item \emph{Store most popular only}: In this case, we set $I = \{1\}$ and $H = \{2,3\}$ and thus, store only the files of the most popular level, level $1$, in the caches.
From \eqref{Eqn:AchievableRateSingleUserPerCache}, the rate of the scheme for this choice is $\max\{N_1 / N_1 - 1, 0\} + K_2 + K_3 = 55$.
\item \emph{Treat all levels as uniform}: In this case, we set $I = \{1,2,3\}$ and $H = \phi$ and thus, allocate equal  memory to all the files.
From \eqref{Eqn:AchievableRateSingleUserPerCache}, the rate of the scheme for this choice is $\max\{(N_1 + N_2 + N_3)/ N_1 - 1, 0\} = 15$.
\item \emph{Merge subset of levels}: Let us set $I = \{1,2\}$ and $H = \{3\}$ and thus, allocate equal memory to all the files belonging to the more popular levels, levels $1$ and $2$.
From \eqref{Eqn:AchievableRateSingleUserPerCache}, the rate of the scheme for this choice is $\max\{(N_1 + N_2)/ N_1 - 1, 0\} + K_3 = 10$.
\end{enumerate}
Thus this example suggests that the optimal choice of $H$ and $I$ is non-trivial and greatly impacts the rate of the proposed scheme.
\end{example}
To understand how to in general choose the sets $H$ and $I$ optimally, consider the following back-of-the-envelope calculation.
Suppose that all levels except one (call it level $\ell$) have been partitioned into two sets $H'$ and $I'$.
If we put $\ell$ with $H'$, we get the achievable rate
\[
R_1(M) \approx \frac{\sum_{i\in I'}N_i}{M} + \sum_{h\in H'}K_h + K_\ell,
\]
whereas if we combine it with $I'$ we get
\[
R_2(M) \approx \frac{\sum_{i\in I'}N_i + N_\ell}{M} + \sum_{h\in H'}K_h.
\]
Then, $R_1(M)\le R_2(M)$ if and only if $K_\ell / N_\ell \le 1/M$, in which case the better choice is to group level $\ell$ with $H'$.
Following this intuition, we choose the sets $H$ and $I$ as:
\begin{equation}
\label{coded:eq:singleuser-hi}
H = \left\{ h\in\{1,\ldots,L\} : K_h / N_h < 1/M\right\};\quad I = \{1,\ldots,L\} \setminus H,
\end{equation}
where as mentioned before, the cache memory is divided equally amongst only the files belonging to levels in $I$ in accordance with the scheme described in Section~\ref{sec:centscheme}.
Recall that for our setup, we can think of the popularity of each file in level $i$ as being proportional to the number of users per file of the level, $K_i/N_i$.
Thus, the above decision rule suggests $1/M$ as a popularity threshold: all files with popularity above this threshold are assigned equal memory and all files with lower popularity are not allocated any memory during the placement phase.

The above scheme leads to the following achievable rate for the single-user multi-level caching setup.
\begin{theorem}
\label{thm:singleuser}
In the single-user setup, the following rate is achievable for all $L$, $K$, $\{N_i,K_i\}$, and $M$:
\[
R_\mathrm{SU}(M) \le \max\left\{ \frac{\sum_{i\in I}N_i}{M} - 1, 0 \right\} + \sum_{h\in H} K_h,
\]
where $H$ and $I$ are as defined in \eqref{coded:eq:singleuser-hi}.
\end{theorem}

As we did for the single-level setup in Section~\ref{sec:approxopt}, next we examine how the performance of the scheme proposed above compares to that of the optimal scheme.
The next result  \cite{HKDmultilevel} states the approximate optimality of the achievable rate $R_{\mathrm{SU}}(M)$ in Theorem~\ref{thm:singleuser} with respect to the optimal rate $R_{\mathrm{SU}}^\star(M)$ for this setup.
\begin{theorem}
\label{coded:thm:singleuser-orderoptimal}
The rate $R_\mathrm{SU}(M)$ achieved in Theorem~\ref{thm:singleuser} for the system with multi-level popularity and a single user per cache is within a constant multiplicative factor of the information-thoeretically optimal rate $R^{\star}_\mathrm{SU}(M)$.
Specifically, for all values of $L$, $K$, $\{N_i,K_i\}$ with $N_i \ge K_i$, and $M$,
\[
1 \le \frac{R_\mathrm{SU}(M)}{R^\star_\mathrm{SU}(M)} \le 72.
\]
\end{theorem}
Note that the bound is independent of the problem parameters.
As before, the proof derives information-theoretic lower bounds on the optimal rate using cut-set based arguments; details are available in  \cite{HKDmultilevel}.
Finally, the focus of the above result is on proving constant factor-optimality (irrespective of system parameters) and while the factor of $72$ is very large, this can be vastly improved using the aforementioned progress made on designing better achievable strategies and lower-bound arguments.

\subsection{Multi-User Setup}

We begin with some notation, see \figurename~\ref{coded:fig:setup-ml-mu} for an illustration of multi-level multi-user setup.
For each popularity level $i$, each cache has exactly $U_i$ users requesting files from level $i$, which implies that the total number of users demanding files from level $i$ is $K_i=KU_i$.
As mentioned earlier, we assume that $N_i \ge K_i=KU_i$ for each level $i$.

As compared to the single-user setup described in the previous section, the biggest difference in the multi-user setup is that every level is represented at each cache, equally across the caches.
In other words, every cache has the same user profile, where a user profile is an indicator of the number of users requesting a file from each given level.
This allows separating the popularity levels and restricting all coding opportunities to be amongst users requesting files from a single level and not across levels.

More precisely, the idea is to partition the memory $M$ among the popularity levels, giving level $i\in\{1,\ldots,L\}$ a memory of $\alpha_iM$ for some $\alpha_i\in[0,1]$, and then apply the single-level coded caching scheme from Section~\ref{coded:sec:overview} on each row of users separately, with each row consisting of users requesting files from a single level.
Under this strategy, we can derive the achievable rate using Theorem~\ref{coded:thm:singlelevel-achieve} and \eqref{coded:eq:singlelevel-ub} to be
\begin{equation}
\label{eqn:multiuserachievable}
R_\mathrm{MU}(M) \le \sum_{i=1}^L U_i \cdot \min\left\{K, \max\left\{ \frac{N_i}{\alpha_iM}-1, 0 \right\} \right\}.
\end{equation}
The factor $U_i$ appears because there are exactly $U_i$ rows of users for level $i$.

By optimizing the overall rate over the memory-sharing parameters $\alpha_1,\ldots,\alpha_L$, we establish a memory allocation which we will show achieves a rate that is information-theoretically order-optimal.
At a high level, this allocation is done by partitioning the popularity levels into three sets: $H$, $I$, and $J$.
The levels in $H$ have such a small popularity that they will get no cache memory.
Thus, for all levels $h\in H$, we will assign $\alpha_hM=0$.
On the opposite end of the spectrum, the most popular levels are assigned to $J$ and are given enough cache memory to completely store all their files in every cache.
Thus, for every level $j\in J$, we have $\alpha_jM=N_j$, since that is the amount of memory needed to completely store all files of level $j$ in each cache.
Finally, the rest of the levels, in the set $I$, will share the remaining memory among themselves, obtaining some non-zero amount of memory per cache but not enough to completely store all of their files in every cache.
The more popular files should get more memory, and as discussed before we can think of $KU_i/N_i$ as representing the popularity of a level $i$.
For the order-optimal strategy we propose, we choose to give level $i$ a memory per cache of  roughly $\alpha_iM\propto N_i\cdot\sqrt{U_i/N_i}$ (hence the memory \emph{per file} is proportional to $\sqrt{U_i/N_i}$).%
\footnote{The square root comes from minimizing the rate expression in \eqref{eqn:multiuserachievable} which has an inverse function of $\{\alpha_i\}$.}

The above assignment will represent a valid choice for the memory-sharing parameters as long as the partition $(H, I, J)$ is selected so that each $\alpha_i \in [0,1]$.
When we plug the above choice of the memory-sharing parameters into \eqref{eqn:multiuserachievable}, we get the following result.
\begin{theorem}
\label{thm:multi-user-achievability}
Given a multi-user caching setup, with $K$ caches, $L$ levels, and, for each level $i$, $N_i$ files and $U_i$ users per cache, and a cache memory of $M$, the following rate%
\footnote{This expression of the rate is a slight approximation that we use here for simplicity as it is more intuitive.
An exact and complete description of the achievable rate can be found in \cite{HKDmultilevel}.}
is achievable:
\begin{equation}
\label{Eqn:UpperBoundMultiUserRate}
R_\mathrm{MU}(M) \approx \sum_{h\in H} KU_h
+ \frac{ \left( \sum_{i\in I} \sqrt{N_iU_i} \right)^2 }{ M - \sum_{j\in J}N_j}
- \sum_{i\in I}U_i,
\end{equation}
where $(H,I,J)$ is the \emph{unique} partition of the set of popularity levels that satisfies:
\begin{IEEEeqnarray*}{lCl/rCcCl}
\forall h &\in& H, &
&& \tilde M &<& \frac1K \sqrt{\frac{N_h}{U_h}};\\
\forall i &\in& I, &
\frac1K \sqrt{\frac{N_i}{U_i}} &\le& \tilde M &\le& \left( 1+\frac1K \right)\sqrt{\frac{N_i}{U_i}};\\
\forall j &\in& J, &
\left( 1 + \frac1K \right)\sqrt{\frac{N_j}{U_j}} &<& \tilde M,
\end{IEEEeqnarray*}
where $\tilde M \approx (M-\sum_{j\in J}N_j)/\sum_{i\in I}\sqrt{N_iU_i}$.
\end{theorem}
The proof of the above result is rather involved and we point the reader to \cite{HKDmultilevel} for details.
Intuitively, since a level $h\in H$ receives no cache memory, all requests from its $KU_h$ users must be handled directly from the broadcast.
Since, we have $N_i\ge K_i = KU_i$ for all levels $i$, then in the worst case a total of $KU_h$ distinct files must be completely transmitted for the users requesting files from level $h$.
This contributes the term $\sum_{h\in H} KU_h$ in the expression of the achievable rate \eqref{Eqn:UpperBoundMultiUserRate}.
The users in set $J$ require no transmission as the files are completely stored in all the caches; however, it does affect the rate through the memory available for levels in $I$.
This is apparent in the expression $M-\sum_{j\in J}N_j$ in \eqref{Eqn:UpperBoundMultiUserRate}.
Finally, the levels in $I$, having received some memory, require a rate that is inversely proportional to the effective memory and that depends on the level-specific parameters $N_i$ and  $U_i$.

Notice in the statement of the theorem that in the inequalities defining the chosen partition $(H,I,J)$, the different sets are largely determined by the quantity $\sqrt{N_i/U_i}$ for each level $i$, which is a function of the file popularities.
Moreover, the inequalities satisfy the natural choice that the most popular levels (i.e., those with the \emph{smallest} $N_i/U_i$) will be in $J$, while the least popular levels (those with the \emph{largest} $N_i/U_i$) will go to the set $H$.

\begin{example}
Consider an example multi-level multi-user setup with $K=10$ caches, $L=3$ file popularity levels, $N_1 = 100$, $N_2 = 200$, $N_3 = 300$ files and $U_1 = 10$, $U_2 = 5$, $U_3 = 1$ users/cache.
Consider the memory per cache to be $M = N_1 = 100$.
We evaluate the rate of the above proposed strategy for different choices of $H$, $I$, $J$:
\begin{enumerate}
\item \emph{Store most popular only}: In this case, we set $J = \{1\}$, $I = \phi$, and $H = \{2,3\}$ and thus, store only the files of the most popular level, level $1$, in the caches.
From \eqref{Eqn:UpperBoundMultiUserRate}, the rate of the scheme for this choice is $KU_2 + KU_3 = 60$.
\item \emph{Share memory amongst all levels}:  In this case, we set $I = \{1,2,3\}$ and thus allocate memory to each level in proportion to the square root of its popularity.
From \eqref{Eqn:UpperBoundMultiUserRate}, the rate of the scheme for this choice is approximately $\frac{ \left( \sqrt{N_1U_1} + \sqrt{N_2U_2} + \sqrt{N_3U_3} \right)^2 }{ N_1 }
- (U_1 + U_2 + U_3) \approx 65 - 16 = 49$.
\item  \emph{Share memory amongst subset of levels}:  In this case, we set $I = \{1,2\}$ and $H=\{3\}$, and thus allocate memory only to the more popular level in proportion to the square root of its popularity.
From \eqref{Eqn:UpperBoundMultiUserRate}, the rate of the scheme for this choice is approximately $KU_3 + \frac{ \left(\sqrt{N_1U_1} + \sqrt{N_2U_2} \right)^2 }{ N_1 }
- (U_1 + U_2)  = 50 - 15 = 35$.
\end{enumerate}
Thus, we say that the optimal choice of $H$, $I$, $J$ is non-trivial and greatly impacts the rate of the proposed scheme.
\end{example}

The next result  \cite{HKDmultilevel} states the approximate optimality of the achievable rate $R_{\mathrm{MU}}(M)$ in Theorem~\ref{coded:thm:singlelevel-achieve} with respect to the optimal rate $R_{\mathrm{MU}}^\star(M)$ for this setup.
\begin{theorem}
\label{coded:thm:multiuser-orderoptimal}
The rate $R_\mathrm{MU}(M)$ achieved in Theorem~\ref{thm:multi-user-achievability} for the system with multi-level popularity and multiple users per cache is within a constant multiplicative factor of the information-thoeretically optimal rate $R^\star_\mathrm{MU}(M)$.
Specifically, for all values of $L$, $K$, $M$, $\{N_i,U_i\}$ with $N_i \ge K_i$ and satisfying regularity condition\footnote{The reasoning behind this condition is that, if it did not hold for some levels $i$ and $j$, then we can think of them as essentially one level with $N_i+N_j$ files and $U_i+U_j$ users per cache.
The resulting popularity $\frac{U_i+U_j}{N_i+N_j}$ would be close to both $U_i/N_i$ and $U_j/N_j$.} $\sqrt{\frac{U_i/N_i}{U_j/N_j}} \ge \frac{1}{\beta}$,
\[
1 \le \frac{R_\mathrm{MU}(M)}{R^\star_\mathrm{MU}(M)} \le c.
\]
where $\beta = 198$ and $c=9909$ are constants (independent of all problem parameters).
\end{theorem}
Unlike the approximate optimality results presented before, the proof of this theorem requires the use of \textit{non cut-set} based arguments to derive information-theoretic lower bounds on the optimal rate; details are available in  \cite{HKDmultilevel}.
As before, the constants involved can all potentially be improved greatly.

%% file: coded/sections/multi-access.tex
So far, we have only considered situations in which each user accesses exactly one cache, with no flexibility.
However, in a wireless heterogeneous network such as the one in \figurename~\ref{coded:fig:hetnet}, the density of access points that have caches could be high enough for each user to potentially access a large number of caches.
This enables some interesting capabilities that can be harnessed to achieve a lower broadcast rate $R$ for the same cache memory $M$.
For instance, each user could have access to the contents of multiple caches at once, effectively increasing the memory available to it.
Alternatively, we could allow the system to adaptively assign to each user one cache out of a set of nearby caches, based on the file that it requested.
In this section, we explore the latter approach in detail as studied in \cite{HKMDadaptive}, and leave the former as a short discussion at the end.

\subsection{Overview of Adaptive User-to-Cache Matching}

In the adaptive matching setup, we keep the restriction of each user accessing the contents of exactly one cache, but allow the flexibility of choosing which cache (possibly among some subset of caches) the user should access based on its requested file. An additional restriction is a load constraint on the caches: each cache can only serve at most one user. Such a problem was studied in \cite{leconte2012, moharir2016}, in the extreme case where all users are able to access any cache.
The surprising insight in both papers is that, contrary to the ``static matching case'' where each user is pre-attached to a unique cache (the setting described in Section~\ref{coded:sec:overview}), an approximately-optimal scheme is to replicate complete files across multiple caches in proportion to their popularity in the placement phase;, and then during the delivery phase, match as many users as possible to a cache that holds its requested file.

We thus observe a dichotomy between two extremes: in the static matching case (when each user is restricted to one cache), appropriate splitting of files and careful placement of subfiles to enable coded-multicast transmissions during delivery as described in Section~\ref{coded:sec:overview} is approximately optimal, while simple file replication is not; on the other hand, in the fully ``adaptive matching case" where each user can be matched to any cache during the delivery phase, appropriate file replication coupled with maximum matching during delivery is approximately optimal, while a static pairing of users and caches along with the coded caching approach of Section~\ref{coded:sec:overview} is sub-optimal.
The natural next question is then: what happens when each user can be matched adaptively to one of a \emph{subset} of caches?
This problem was studied in \cite{HKMDadaptive}, and we will discuss its main results here.

\subsection{System Model}
\label{Sec:sysmodeladapmat}
\begin{figure}
\centering
\includegraphics[scale=\codedsetupfigsscale]{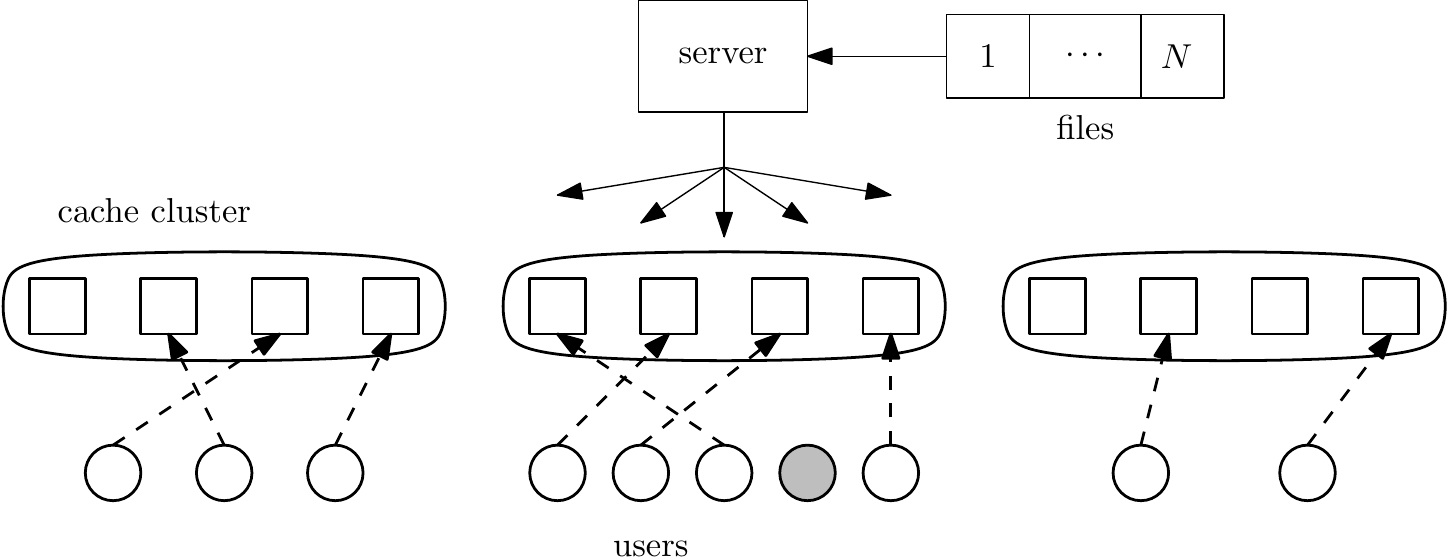}
\caption{Partial adaptive matching setup.
The caches are partitioned into clusters, and each user can be matched to one cache in its cluster, with a load constraint on the caches.
Excess users in a cluster cannot be matched, such as the user colored in grey.}
\label{coded:fig:adaptive-setup}
\end{figure}

Suppose there is a content library of $N$ files, called $W_1,\ldots,W_N$.
There are $K$ caches, partitioned into $K/d$ mutually exclusive \emph{clusters} of $d$ caches each (assume $d$ divides $K$).
At each cluster $c$, there is a stochastic number of users $u_n(c)$ that request file $W_n$, where $u_n(c)$ is a Poisson random variable of parameter $\rho d/N$, with $\rho\in(0,1/2)$ a constant.
Thus at every cluster (whose size is $d$ caches), the expected number of users is $\rho d$.
We will refer to $\mathbf{u}=\{u_n(c)\}_{n,c}$ as the \emph{user profile}.

In addition to the usual placement and delivery phases, there is an intermediate \emph{matching phase}, which occurs after the users have made their requests.
In this phase, we assign each user in a cluster to one cache in the same cluster, subject to a load constraint of no more than one user per cache.
Thus if there are more users than caches in a cluster ($\sum_n u_n(c)>d$ for some $c$), there will necessarily be some unmatched users who will have access to the contents of no cache.
Note that the placement phase occurs without knowing the user profile, while both the matching phase and the delivery phase have  knowledge of the user profile.

Let $R_\mathbf{u}$ denote the broadcast rate given a specific user profile $\mathbf{u}$.
We are interested in the expected rate $\bar R = \mathbb{E}_\mathbf{u}[R_\mathbf{u}]$, and more specifically  the optimal expected rate $\bar R^\star(M)$ for every memory $M$ over all possible placement, matching, and delivery strategies.

The choice of a Poisson number of users is useful as it not only more closely models real-world user requests, but also simplifies the analysis in this problem.
There is also little difference, fundamentally, between the Poisson model and the model with a fixed number of users (such as the one studied in the previous sections), as long as the cluster size $d$ is large enough, namely $d=\Omega(\log K)$.
This means that comparisons with other works in the literature are possible.
Note that for smaller $d$, the Poisson model is less meaningful; when $d=1$ for instance, there is positive probability for each cluster to have more than one user, which means that with high probability, a significant fraction of the users cannot be matched to any cache and resultantly a high server transmission rate is necessary irrespective of the cache memory size.

As mentioned above, the Poisson model only makes sense for $d=\Omega(\log K)$, and so we adopt this regularity condition in this section.
More precisely, we assume that
\begin{equation}
\label{coded:eq:d-ge-logk}
d \ge \frac{2(1+t_0)}{\alpha} \log K,
\end{equation}
where $\alpha=-\log(2\rho e^{1-2\rho})>0$, and $t_0>0$ is some constant.
Finally, we restrict our attention to the case when $N\ge K$.

\subsection{Balancing Two Extremes}

The model described above, known as the partial adaptive matching setup, is a generalization of the two extremes.
When $d=1$, we have a static matching setup as in \cite{maddah-ali2012} (while the Poisson model is not meaningful here, insights can still be gained).
When $d=K$, we have the full adaptive matching setup as in \cite{leconte2012, moharir2016}.

As discussed above, there is a dichotomy between these two extremes: the former favors a coded delivery scheme, while the other favors an uncoded replication scheme.
In what follows, we examine how each scheme performs if adapted to the partial adaptive matching setup.
Specifically, we look at:
\begin{itemize}
\item Pure Coded Delivery (PCD): ignores any potential adaptive matching benefits by arbitrarily assigning users to caches, and applying a standard Maddah-Ali--Niesen scheme as discussed in Section~\ref{coded:sec:overview};
\item Pure Adaptive Matching (PAM): ignores any potential coding gains and focuses only on file replication within a cluster and on adaptively matching users to caches within a cluster.
\end{itemize}

As we will see, in the general case we observe two regimes, and each scheme will be preferred in one regime.
These regimes are roughly defined by a threshold on the \emph{total cluster memory} $dM$: when $dM\ll N$ then PCD is favorable, and when $dM\gg N$ then PAM is favorable.
Furthermore, in each regime, the favorable scheme is approximately optimal for almost all values of the cache memory.
This is illustrated in \figurename~\ref{coded:fig:pam-pcd}.
Notice that in the special case $d=1$ (respectively, $d=K$), \figurename~\ref{coded:fig:pam-pcd} shows that PCD (respectively, PAM) is always preferred, as expected from the previous results.

\begin{figure}
\centering
\includegraphics[width=.5\textwidth]{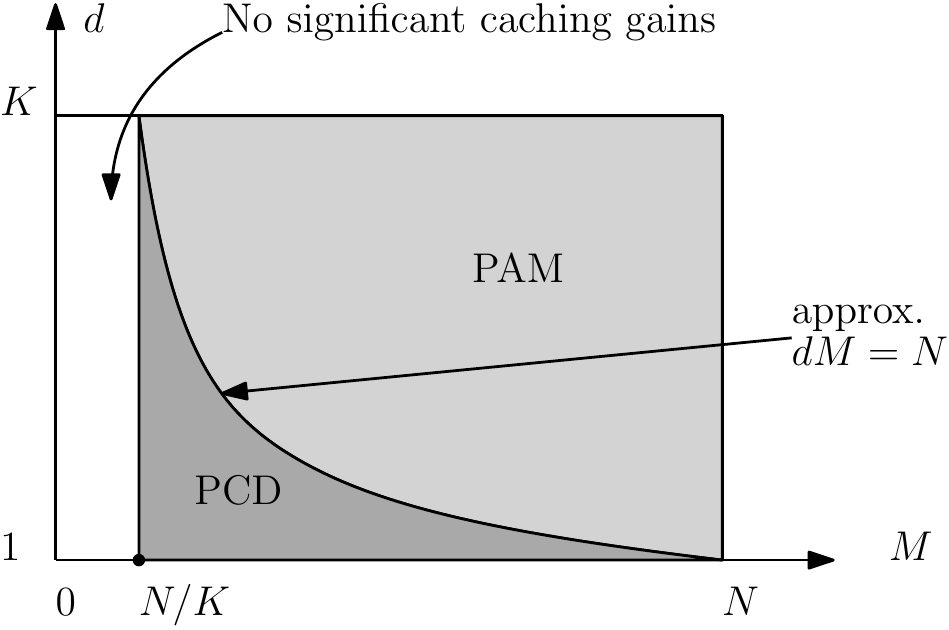}
\caption{An approximate visualization of the regimes in which each scheme is more favorable.
The boundary between the PCD- and PAM-dominated regions is blurry: the regime $\Omega(N)<dM<O(N\log N)$ is still not very well understood.}
\label{coded:fig:pam-pcd}
\end{figure}

Each of the two schemes focuses on one idea: PCD ignores adaptive matching in favor of coding gains, and PAM ignores coded delivery in favor of adaptive matching gains.
A hybrid coding and matching (HCM) scheme is introduced in \cite{HKMDadaptive} which performs better than both schemes in most memory regimes.

\subsection{The Pure Coded Delivery (PCD) scheme}

The PCD scheme is a straightforward adaptation of the Maddah-Ali--Niesen scheme described in Section~\ref{coded:sec:overview}; the placement phase is identical to the one  described there. 
During the matching phase, we pick any valid user-to-cache matching and provide each user with access to the corresponding cache; this is sufficient since the placement is completely symmetric with respect to the caches and the files. The delivery phase is conducted in two parts:
\begin{enumerate}
\item For the subset of users which were matched to caches, delivery proceeds in the same fashion as in Section~\ref{coded:sec:overview} by creating coded-multicast transmissions.
\item Any users that were not matched (because there were more users than caches in their cluster) will simply be served directly by the server. 
\end{enumerate}

Note that for the model described in Section~\ref{Sec:sysmodeladapmat}, the expected number of such unmatched users is very small.
In fact it can be shown that
\[
\mathbb{E}[U^0] \le K^{-t_0}/\sqrt{2\pi},
\]
where $U^0$ is the total number of excess users across all clusters and $t_0 > 0$ is a  positive constant. Note that the expected number of excess users goes to zero as $K$ increases. All the other users (i.e., those that are matched to some cache) will be served by the basic Maddah-Ali--Niesen scheme, and so PCD can achieve the rate in the following theorem.

\begin{theorem}
\label{thm:pcdrate}
For the partial adaptive matching model described in Section~\ref{Sec:sysmodeladapmat}, the expected rate achieved by PCD is
\[
\bar R^\mathrm{PCD}(M) \le \min\left\{ \rho d, \left[\frac{N}{M} - 1\right]^+ + \frac{K^{-t_0}}{\sqrt{2\pi}} \right\}.
\]
\end{theorem}

Notice that this is not much different to the rate achieved in the static matching setup.
This is expected since we are not making any intelligent use of the adaptive matching feature at all in PCD. However, this turns out to be approximately optimal when the total  memory in any  cluster is not enough to hold the entire library, as stated next.

\begin{theorem}
\label{thm:PCDopt}
When $M\le(1-e^{-1}/2)N/2d$, the expected rate achieved by PCD is approximately optimal in the sense that
\[
\bar R^\mathrm{PCD}(M) \le C \cdot \bar R^\star(M) + o(1),
\]
where $\bar R^\star(M)$ is the information-theoretically optimal rate, $C$ is a constant independent of the problem parameters, and the $o(\cdot)$ notation is to be understood with respect to the growth of $K$.
\end{theorem}

We skip the proof of these results here and instead point the interested reader to  \cite{HKMDadaptive}, which has all the details as well as discusses more general scenarios with non-uniform content popularity.

\subsection{The Pure Adaptive Matching (PAM) scheme}

As previously mentioned, the PAM scheme takes the opposite approach to PCD.
It ignores all possible coding in favor of a more intelligent matching of users to caches.
The idea is to store only replicas of files in every cluster, and rely as much as possible on the matching phase to connect each user to a cache that contains the file that it requested.

More precisely, the three phases work as follows.
In the placement phase,  the total cluster memory is $dM$ and we store a complete copy of every file in $\lfloor dM/N\rfloor$ caches in every cluster.
In the matching phase, we find the best matching of users to caches so that the number of users matched to a cache containing their requested file is maximized.
In the delivery phase, any users that could not be successfully matched to a suitable cache are served directly from the server.

Notice that the scheme only really takes off once $dM\ge N$: for smaller memory values, there is a significant fraction of users whose requests can be satisfied locally and have to be served directly by the server.
What's more interesting is that after this threshold of $dM\ge N$, the achieved expected rate decays exponentially with the cluster memory!
The precise rate expression is given in the following theorem.

\begin{theorem}
For the partial adaptive matching model described in Section~\ref{Sec:sysmodeladapmat},  the expected rate achieved by PAM is
\[
\bar R^\mathrm{PAM}(M) \le \begin{cases}
\rho K & \text{if $M<N/d$;}\\
KMe^{-\rho h dM/N} & \text{if $M\ge N/d$,}
\end{cases}
\]
where $h=(1/\rho)\log(1/\rho)+1-1/\rho$.
\end{theorem}

The proof follows along similar lines as \cite{leconte2012}, which focuses on the fully adaptive matching case, and generalizes the results to the partially adaptive matching case.

Notice that $\bar R^\mathrm{PAM}(M)=o(1)$ when $dM>\Omega(N\log N)$. Thus, once the total cluster memory is slightly larger than the total catalogue size, the PAM scheme requires negligible server transmission rate. This also trivially implies that PAM is approximately information-theoretically optimal in that regime.
Combining with Theorem~\ref{thm:PCDopt}, we have that PCD is approximately optimal when $dM<O(N)$ and PAM is approximately optimal when $dM>\Omega(N\log N)$, as illustrated approximately  in \figurename~\ref{coded:fig:pam-pcd} (ignoring the $\log N$ factor for simplicity).

\subsection{The Hybrid Coding and Matching (HCM) scheme}

So far, we have looked at two schemes that each focuses on a single gain: either a coding gain or an adaptive matching gain.
The schemes are approximately optimal in complementary regimes, as illustrated in \figurename~\ref{coded:fig:pam-pcd}.
This section explores a hybrid scheme that unifies coded delivery and adaptive matching by incorporating ideas of both PCD and PAM.
This Hybrid Coding and Matching (HCM) scheme, first introduced in \cite{HKMDadaptive}, turns out to perform better than both PCD and PAM in most memory regimes.

The hybrid scheme combines ideas from both PCD and PAM by introducing a coloring scheme at both the cache level and the file level.
First, we choose a certain number of colors $\chi\in\{1,\ldots,d\}$.
The exact value is not important for now; the ideas work for any $\chi$.
We then partition the caches in every cluster into $\chi$ subsets of (almost) equal size, and color each subset with a unique color.
Similarly, we partition the files in the content library into $\chi$ subsets of (almost) equal size, and color each subset with one color.
Finally, we apply the coded delivery ideas \emph{within each color}, while applying the adaptive matching ideas \emph{across colors}.

More precisely, the three phases proceed as follows.
In the placement phase, for every color $x$ we perform a Maddah-Ali--Niesen placement of the files \emph{of color $x$} in only the caches of the same color.
The placement phase is agnostic to the cluster to which a cache belongs.
In the matching phase, every user can be matched to an arbitrary cache in its' cluster whose color matches the file that the user requested; the user is matched to the color, but the choice of cache within that color is arbitrary. 
In the delivery phase, the Maddah-Ali--Niesen coded delivery is performed for every color $x$ separately, and unmatched users are served directly.
Like the placement phase, the delivery phase ignores clusters and serves all users of the same color in the same broadcast message.

Since we have $\chi$ sub-systems of $N/\chi$ files, each running a separate Maddah-Ali--Niesen scheme, using Theorem~\ref{thm:pcdrate} we can show that the hybrid scheme can achieve a rate of
\[
\bar R(M) \approx \min\left\{ \rho K, \chi \cdot \left( \frac{N/\chi}{M} - 1 \right) + \bar U^0(\chi) \right\},
\]
where $\bar U^0(\chi)$ is the expected number of unmatched users when choosing $\chi$ colors.
What is left is therefore to choose the right value of $\chi$.

If the number of colors is too small, then there is little benefit in adaptive matching since the number of choices is reduced.
Conversely, if the number of colors is too large, then the number of caches in each color becomes small, and it becomes likely that a significant number of colors have fewer caches than there are users requesting a file from them; the number of unmatched users thus becomes too large.
The balance is struck when $\chi\approx d/\log K$ colors, as stated more precisely in the next theorem.

\begin{theorem}
For any $t\in[0,t_0]$, the HCM scheme can achieve an expected rate of
\[
\bar R^\mathrm{HCM}(M) \le \begin{cases}
\min\left\{ \rho K, \frac{N}{M} - \chi + \frac{K^{-t}}{\sqrt{2\pi}} \right\} & \text{if $M\le\lfloor N/\chi\rfloor$;}\\
\frac{K^{-t}}{\sqrt{2\pi}} & \text{if $M\ge\lceil N/\chi\rceil$,}
\end{cases}
\]
where $\chi = \lfloor \alpha d/(2(1+t)\log K)\rfloor$ and $t_0 > 0$ is a positive constant.
\end{theorem}

The rate expression can be approximately written as
\[
\bar R^\mathrm{HCM}(M) \approx \min\left\{ \rho K,
\left[ \frac{N}{M} - \Theta\left(\frac{d}{\log K}\right) \right]^+
+ o(1) \right\}.
\]

Comparing the performances of PCD, PAM, and HCM, we find that HCM performs better than both of them in most memory regimes.
In fact, HCM is a unified scheme that is approximately optimal for almost all memory regimes.
Specifically, we have:
\begin{itemize}
\item For all $M\ge0$, HCM performs better than PCD;
\item When $dM\le O(N)$, both HCM and PCD are approximately optimal, while PAM is not;
\item When $dM\ge\Omega(N\log N)$, both HCM and PAM achieve a rate of $o(1)$ and are trivially approximately optimal, while PCD is not;
\item The intermediate regime $\Omega(N)\le dM\le O(N\log N)$ is not very well understood and the exact relationship between the different rates, as well as their approximate optimality, is not known.
\end{itemize}
More details can be found in \cite{HKMDadaptive}.

\subsection{Simultaneous Cache Multi-Access}
\input{coded/sections/simultaneous-access.tex}

%% file: coded/sections/simultaneous-access.tex
In the previous section, we studied the adaptive matching setup where  each user has many nearby caches but based on its file request is only matched to one among them. In this section, we will briefly discuss an alternative setting where each user is allowed access to the information stored in \emph{all} the neighboring caches. This problem was introduced in \cite{HKDmultilevel} where it was analyzed within the larger multi-level popularity setting.
In this section, we will restrict the discussion to a uniform popularities setup in order to focus on the simultaneous multi-access aspect of the problem.

The first question we must ask ourselves is: what sort of multi-access model should we adopt here?
A setting with caches divided into clusters like in Section~\ref{Sec:sysmodeladapmat} is not very interesting in this scenario.
Indeed, suppose like in Section~\ref{Sec:sysmodeladapmat}  that the caches are partitioned into clusters of $d$ caches each, and that every user could access all the caches in its' cluster. Thus, any two users in the same cluster will have access to the same subset of caches. Then the problem is effectively reduced to the basic setup seen in Section~\ref{sec:setupbasic}, but with $K/d$ caches of memory $dM$ each and multiple users accessing each cache.
This is a special case of the multi-level multi-user setup described in Section~\ref{coded:sec:popularity}, restricted to a single popularity level. Hence a cache-cluster model with simultaneous cache-access is not very interesting.

A more interesting scenario is when the sets of caches that users can access have non-trivial intersections.
In other words, two users can have a few caches in common, but also a few caches that the other does not have access to.
One way to model this is using a ``sliding window'' approach, where user $k$ accesses caches $k,k+1,\ldots,k+d-1$ for some $d \in \{1,2,\ldots,K\}$, using a cyclic wrap-around to preserve symmetry.
Specifically, if we label the caches as $Z_1$ through $Z_K$, then user $k\in\{1,\ldots,K\}$ has access to the $d$ caches
\[
Z_k, Z_{\langle k+1\rangle}, \ldots, Z_{\langle k+d-1\rangle},
\]
where $\langle m\rangle=m$ if $m\le K$ and $\langle m\rangle=m-K$ if $m>K$.
Thus if $K=4$ and $d=2$, then user $1$ has access to caches $Z_1$ and $Z_2$, while user $4$ accesses caches $Z_4$ and $Z_1$.

This problem setup can be motivated by a scenario in which caches are arranged linearly and users access the $d$ nearest caches to them.
While this linearity assumption is simplistic, the problem can be easily extended to a more realistic scenario in which the caches are arranged in a two-dimensional lattice and as before every user accesses the $d$ nearest caches.

At this point, it is interesting to think about how the local and global caching gains would be different in this scenario compared with the basic setup in Section~\ref{sec:setupbasic}.
Recall that the global caching gain is caused by the total memory in the system, $KM$.
In this scenario, the total memory is still $KM$, so we might not expect the global caching gain to be different.
However, also recall that the local caching gain is caused by the cache memory available for each user, which in the basic setup was $M$.
But a key difference in this simultaneous multi-access problem is that every user actually has access to a memory of $dM$.
Thus one might expect that in the simultaneous multi-access problem we can achieve a rate of
\begin{equation}
\label{coded:eq:simultaneous-caching-gains}
R_\mathrm{SM}(M) \approx K \cdot \left( 1 - \frac{dM}{N} \right) \cdot \frac{1}{1+KM/N},
\end{equation}
which is similar to the expression in Theorem~\ref{coded:thm:singlelevel-achieve} except for the $dM$ term in the factor that represents the local caching gain. Note that it is not immediate whether the above rate expression is achievable for the setup being considered here, since no two users share the same cache-access structure.  

We will now analyze \eqref{coded:eq:simultaneous-caching-gains} in order to get insights into schemes that can achieve this rate.
Notice that the effect of multi-access only appears when $dM$ is larger than some fraction of $N$, e.g., $dM>N/2$.
Below this threshold, the global caching gain, which is not affected by multi-access, dominates.
This inspires the following simple scheme:
\begin{itemize}
\item When $dM\le N/2$, we ignore multi-access and assume that user $k$ only accesses cache $Z_k$. We apply the Maddah-Ali--Niesen scheme under this assumption, achieving a rate of $R(M)\le\min\{K,N/M\}$.
\item When $dM=N$, apply a $(K,d)$-erasure-correcting code on each file, creating $K$ coded messages of size $F/d$ bits each, such that any $d$ of them can recreate the entire file.
We store each such coded message in one unique cache.
Thus every user, by accessing the $d$ caches in its neighborhood, can recover any file by retrieving the corresponding $d$ coded messages in those caches.
This achieves a rate of zero.
\item When $N/2<dM<N$, we use memory sharing between the two schemes at $dM=N/2$ and $dM=N$ respectively, to achieve a linear combination of the two rates.
\end{itemize}
The scheme described above achieves the rate expression in the theorem below.
\begin{theorem}
\label{coded:thm:simultaneous-multi-access}
In the simultaneous multi-access problem with $N$ files, $K$ caches and users, and a cyclic cache-access structure with a per user access degree of $d$, we can achieve a rate of
\[
R_\mathrm{SA}(M) \le 4 \cdot \min\left\{K, \frac{N}{M}\right\}\left(1 - \frac{dM}{N}\right),
\]
for all cache memory $M\in[0,N/d]$.
A rate of zero is achieved for $M\ge N/d$. Furthermore, the gap of the achievable rate to the information-theoretically optimal rate $R_\mathrm{SA}^\star(M)$ is given by
\[
1 \le \frac{R_\mathrm{SA}(M)}{R_\mathrm{SA}^\star(M)} \le c\cdot d,
\]
where $c$ is some constant.
\end{theorem}

The proof of Theorem~\ref{coded:thm:simultaneous-multi-access} and further details can be found in \cite{HKDmultilevel}.

%% file: coded/sections/Generalizations.tex
In the previous sections, we have studied a simple network structure where the server communicates directly with the users via an error-free broadcast link. We here briefly describe some ways in which this aspect of the setup has been generalized in the literature. 

\subsection{Network Topologies}
There have been several works in the literature which study more complex topologies for the cache network. We briefly describe two such models below:
\begin{enumerate}
\item \emph{Hierarchical networks}: In this model, the server is connected to the users via a tree network, with caches of possibly different sizes at each level of the tree and where each cache at level $i$ communicates with its' children at level $i+1$ via an error-free broadcast link. Note that the server is the root of this hierarchical caching network and the users are the leaves. \cite{KNMDhierarchical} studied the special case of a two-level hierarchical tree caching network with $N$ files of size $F$ bits each at the server, communicating via an error-free broadcast link with $K_1$ \emph{mirrors} each with a cache of memory size $M_1F$ bits, at level $1$. Each of these mirror nodes is connected to $K_2$ \emph{users} each with a cache of size $M_2F$ bits. The system operates as before in two phases: a placement phase when all caches are populated with content, and then after the user requests are revealed, a delivery phase where the server sends a common message of size $R_1F$ bits to the mirrors and each mirror sends a message of size $R_2F$ bits to its connected users. \cite{KNMDhierarchical} proposed a scheme and showed  that for any $M_1, M_2$, the required rates $R_1, R_2$ for the proposed scheme are within a constant factor (independent of all problem parameters) of the information-theoretic optimal rates. A desired feature of the proposed scheme is that the delivery phase only uses messages which involve coding across a single layer of storage at a time. Details can be found in \cite{KNMDhierarchical}. \\
\item \emph{Device-to-device networks}: In this model, there is no designated server in the network which hosts the entire catalogue of $N$ files. The system consists of $K$ co-located users, each with a cache of size $MF$ bits, which can communicate with each other over an error-free broadcast link. This setup was studied in \cite{MingyueD2D1, MingyueD2D2} where they proposed a caching and delivery scheme for this setup and analyzed the total required transmission size on the shared link, as well as compared it to information-theoretic lower bounds. 
\end{enumerate}
%

%
\subsection{Interference networks and physical-layer considerations}
Most of our discussion so far has been limited to a single base station.
An interesting problem is to consider multiple base stations, each of which has access to the content library one way or another.
These base stations will then interfere with each other, and the question becomes how to manage this interference for the purpose of content distribution.

Our discussion has so far ignored the physical layer, instead treating all channels as error-free bit pipes.
In broadcast networks this is not a big issue since separating what to send from how to transmit it is natural in a broadcast setting.
However, this is no longer straightforward when several base stations are present.

This problem of caching in interference networks was first studied in \cite{MNinterference}, which had an interference channel with three transmitters that were equipped with caches and three receivers that were requesting content.
This was later extended in \cite{TSfogran} to consider an arbitrary number of transmitters and receivers.
Problems with caches both at the transmitters and at the receivers were then studied, but with restrictions on the schemes: \cite{NMAinterferencemanagement} was limited to one-shot linear schemes while \cite{XLTinterference} prohibited coding across files during the placement phase.
Furthermore, both \cite{XLTinterference} and \cite{HNDlayered} looked at a limited number of transmitters and receivers (no more than three).

The first general result was published in \cite{HNDdofcaching}, which found an approximate characterization of the information-theoretically optimal rate-memory trade-off in the high-SNR regime.
Three key insights into the problem are derived.
First, it is shown that a separation of the physical and network layers is approximately optimal: a physical-layer scheme focuses on transmitting some message set across the interference network using a technique known as interference alignment, and the network-layer scheme uses this message set as error-free bit pipes to implement a coded caching scheme.
Second, it is shown that, as long as the transmitters can collectively store exactly the entire content library, then increasing the transmitter memory has no effect on the optimal rate beyond a constant multiplicative factor.
A consequence of this is that it is not necessary for the transmitters to share information: they can all store distinct parts of the content library for most gains to be obtained.
Third, there is a trade-off between the receiver memory and the number of transmitters needed to approximately achieve maximal system performance: as the receiver memory increases, fewer transmitters are required.
In particular, when each receiver can hold a fraction of the library, then a constant number of transmitters is sufficient to achieve most benefits.
We will discuss these results in more detail below.

At the other extreme, the low-SNR regime was studied in \cite{HNDenergyefficiency}, where a similar separation of the network and physical layers is proposed.
This separation architecture is shown to be approximately optimal in some cases, namely the single-receiver and the single-transmitter cases.
Contrary to the high-SNR regime, it is shown in \cite{HNDenergyefficiency} that transmitter co-operation, by storing shared content in their caches, is crucial in the low-SNR regime.

\subsubsection{The Separation Architecture}

Although the separation architecture was studied specifically for a Gaussian interference network, we will first describe it in a very general context as the same ideas hold.
We will then show how it applies to the Gaussian network.

\paragraph{High-Level Overview of the Separation Architecture}

There is a content library containing $N$ files of size $F$ bits.
The library is separated from the users by an interference channel.
The interference channel has $K_t$ transmitters and $K_r$ receivers, who act as the users.
Each transmitter has a cache of memory $M_tF$ bits and each receiver has a cache of memory $M_rF$ bits.

During the placement phase, we place information about the files in every transmitter and receiver cache.
During the delivery phase, each transmitter $\ell\in\{1,\ldots,K_t\}$ sends a codeword $\mathbf{x}_\ell = (x_\ell(1), \ldots, x_\ell(T))$ over a block length of $T$ through the interference network.
Importantly, the codeword $\mathbf{x}_\ell$ can only depend on the file requests and the contents of transmitter $\ell$'s cache; the transmitter has no knowledge of the entire content library.
Each receiver $k\in\{1,\ldots, K_r\}$ then receives a signal $\mathbf{y}_k=(y_k(1),\ldots,y_k(T))$.
Finally, each receiver $k$ uses the received signal $\mathbf{y}_k$ in combination with the contents of its cache to recover the requested file.
The goal is to maximize the transmission rate defined as $R=F/T$.

A key aspect of this problem that has not been discussed in previous sections is that it combines caching with physical-layer delivery.
The main idea of the separation architecture is to separate the caching aspect (what to store in the caches and what to send from transmitters to receivers) from the physical-layer aspect (how to send it through the interference network).
Thus the system is split into an overlay network layer and a physical layer.
The two layers interface through a set of messages from (subsets of) transmitters to (subsets of) receivers.
This is illustrated in \figurename~\ref{coded:fig:separation}.

\begin{figure}
\centering
\includegraphics[width=\textwidth]{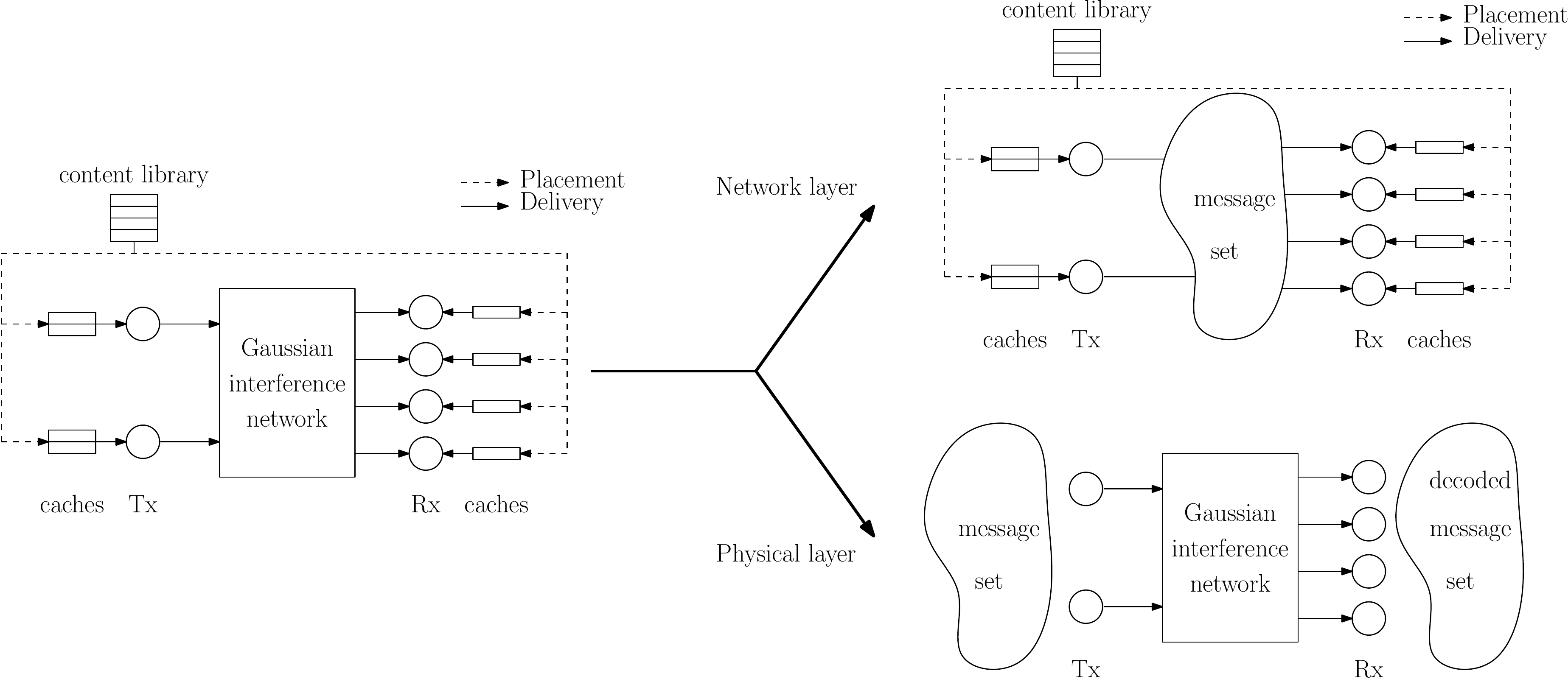}
\caption{Illustration of the separation architecture.
The cache-aided interference channel is split into a physical layer and a network layer; the two layers interface using a message set.}
\label{coded:fig:separation}
\end{figure}

In the most general sense, we define a set of $n$ messages
\[
\mathscr{V} = \left\{ V_{\mathcal{K}_1\mathcal{L}_1}, V_{\mathcal{K}_2\mathcal{L}_2}, \ldots, V_{\mathcal{K}_n\mathcal{L}_n} \right\},
\]
where $\mathcal{K}_i\subseteq\{1,\ldots,K_r\}$ and $\mathcal{L}\subseteq\{1,\ldots,K_t\}$, and $V_{\mathcal{K}_i\mathcal{L}_i}$ is a message from the transmitters in $\mathcal{L}_i$ to the receivers in $\mathcal{K}_i$.
At the physical layer, we have the sub-problem of transmitting this message set across the interference network reliably.
It is assumed that all the transmitters in $\mathcal{L}_i$ have access to the message $V_{\mathcal{K}_i\mathcal{L}_i}$.
At the network layer, we can use these messages as orthogonal, non-interacting bit pipes through which we can pass information from the transmitters to the receivers.
The constraint is that every transmitter in $\mathcal{L}_i$ must be able to cosntruct $V_{\mathcal{K}_i\mathcal{L}_i}$ from the contents of its cache.

Suppose the physical-layer scheme can transmit every message in the message set at a rate of at least $R'$, over a block length of $T$.
Suppose also that the network-layer scheme can deliver the requested files using at most $vF$ bits through each bit pipe represented by the messages in the message set.
We thus have $vF\ge R'T$.
Then, the system would have delivered the requested files to the user at a total rate of $R = F/T \ge v/R'$.

\paragraph{The High-SNR Gaussian Interference Network}

In the particular case of the memoryless Gaussian interference network, the inputs and outputs to the channel are real-valued.
At each time step $\tau$, every output symbol $y_k(\tau)$ is a linear combination of all the input symbols $x_k(\tau)$, plus a Gaussian unit-variance random noise variable.
Furthermore, a power limit of $P$ is imposed on the input codewords, $\|\mathbf{x}_\ell\|^2\le PT$.

We are interested in the high-SNR regime, i.e., the regime in which $P$ is large.
More specifically, we will look at the degrees of freedom (DoF) of the system, which is the behavior of the rate as a scaling of the capacity of a point-to-point Gaussian channel.
Specifically, if we write the information-theoretically optimal rate for a specific power $P$ as $R^\star(P)$, then the degrees of freedom is defined as
\[
\mathsf{DoF} = \lim_{P\to\infty} \frac{R^\star(P)}{\frac12\log P}.
\]

We will next describe the approximately optimal strategy within the context of the high-SNR cache-aided Gaussian interference network.
This strategy makes use of the separation architecture in the following way.
\begin{enumerate}
\item The message set $\mathscr{V}$ that is chosen is a set of messages from every single receiver $\ell$ to every subset $\mathcal{K}$ of receivers of a fixed size $|\mathcal{K}|=\kappa+1$, where $\kappa\approx K_rM_r/N$.
Notice that this is exactly the same as the multicast size in the broadcast setup in Section~\ref{coded:sec:overview}.
Thus the message set represents a set of single-transmitter multicast channels.
\item At the network layer, we partition every file in the content library into $K_t$ parts and store each part at one transmitter.
Thus each transmiter has a content sub-library consisting of part of every file.
A standard centralized Maddah-Ali--Niesen scheme is performed on each sub-library, and each multicast message (intended for $\kappa+1$ receivers) is sent through the corresponding bit pipe.
\item At the physical layer, we apply a technique known as interference alignment in order to transmit the message set as efficiently as possible.%
\footnote{In fact, interference alignment can exactly achieve the degrees of freedom of the communication problem that arises at the physical layer.
Note that this is the degrees of freedom associated with the rate $R'$ described earlier; it is not the degrees of freedom of the entire cache-aided interference network.}
\end{enumerate}

The following theorem gives the approximate degrees of freedom of the network, as determined in \cite{HNDdofcaching}.
\begin{theorem}
\label{coded:thm:dof}
The degrees of freedom of the cache-aided Gaussian interference network is approximately given by
\[
\mathsf{DoF} \approx
\frac{K_tK_r}{K_t+K_r-1}
\cdot \frac{1}{1-M_r/N}
\cdot \frac{K_rM_r/N+1}{\frac{M_r}{N}(\frac{1}{K_r}+\frac{1}{K_t-1})^{-1}+1},
\]
for all $N$, $K_t$, $K_r$, $M_r\in[0,N]$, and $M_t\ge N/K_t$.
The approximation is within a constant multiplicative factor.
\end{theorem}

Notice that the degrees of freedom can be written as the product of three gains, in a similar way to the rate expression in Section~\ref{coded:sec:overview}.
The first term is the interference alignment gain and represents the DoF when no receiver caches are present.
The second term is the local caching gain, as with the broadcast case.
The third term is the global caching gain.

Finally, some powerful insights can be gained from Theorem~\ref{coded:thm:dof}.
The first is that a separation of the network and physical layers is approximately optimal.
Second, one can see that the DoF expression does not involve the transmitter memory $M_t$, excet in the condition $K_tM_t\ge N$.
This implies that there is no benefit (no more than a constant factor) in increasing the transmitter memory.
In particular, it shows that transmitter co-operation is unnecessary: the DoF can be approximately achieved even if the transmitters share no information at all.
Third, as the receiver memory increases, the number of transmitters needed for approximately achieving the DoF decreases.
These insights together show that the cache-aided interference network problem can be solved with a system that is layered and simple to design and implement.